\documentclass[dvipsnames,sigplan,screen,authorversion]{acmart}
\setcopyright{acmlicensed}
\acmPrice{15.00}
\acmDOI{10.1145/3453483.3454088}
\acmYear{2021}
\copyrightyear{2021}
\acmSubmissionID{pldi21main-p440-p}
\acmISBN{978-1-4503-8391-2/21/06}
\acmConference[PLDI '21]{Proceedings of the 42nd ACM SIGPLAN International Conference on Programming Language Design and Implementation}{June 20--25, 2021}{Virtual, Canada}
\acmBooktitle{Proceedings of the 42nd ACM SIGPLAN International Conference on Programming Language Design and Implementation (PLDI '21), June 20--25, 2021, Virtual, Canada}

\usepackage{xspace}
\usepackage{framed}
\usepackage{amsthm}
\usepackage{amsmath}
\usepackage{amsfonts}
\usepackage{booktabs}
\usepackage{microtype}
\usepackage{relsize}
\usepackage{xcolor}
\usepackage{thmtools}
\usepackage{listings}
\usepackage{enumitem}
\setlist[itemize]{leftmargin=2em}

\definecolor{dkcyan}{rgb}{0.1, 0.3, 0.3}
\definecolor{dkgreen}{rgb}{0,0.3,0}
\definecolor{olive}{rgb}{0.5, 0.5, 0.0}
\definecolor{dkblue}{rgb}{0,0.1,0.5}

\definecolor{col:ln}{rgb}  {0.1, 0.1, 0.7}
\definecolor{col:str}{rgb} {0.8, 0.0, 0.0}
\definecolor{col:db}{rgb}  {0.9, 0.5, 0.0}
\definecolor{col:ours}{rgb}{0.0, 0.7, 0.0}

\definecolor{lightgreen}{RGB}{170, 255, 220}
\definecolor{darkbrown}{RGB}{121,37,0}

\colorlet{listing-comment}{gray}
\colorlet{operator-color}{darkbrown}

\lstdefinelanguage{custom-haskell}{
    language=Haskell,
    deletekeywords={lookup, delete, map, mapMaybe, Ord, Maybe, String, Just, Nothing, Int, Bool},
    keywordstyle=[2]\color{dkgreen},
    morekeywords=[2]{String, Map, Ord, Maybe, Int, Bool},
    morekeywords=[2]{Name, Expression, ESummary, PosTree, Structure, HashCode, VarMap},
    keywordstyle=[3]\color{dkcyan},
    literate=%
        {=}{{{\color{operator-color}=}}}1
        {|}{{{\color{operator-color}|}}}1
        {\\}{{{\color{operator-color}\textbackslash$\,\!$}}}1
        {.}{{{\color{operator-color}.}}}1
        {=>}{{{\color{operator-color}=>}}}1
        {->}{{{\color{operator-color}->}}}1
        {<-}{{{\color{operator-color}<-}}}1
        {::}{{{\color{operator-color}::}}}1
}

\lstdefinestyle{default}{
    basicstyle=\ttfamily\fontsize{8.7}{9.5}\selectfont,
    columns=fullflexible,
    commentstyle=\sffamily\color{black!50!white},
    escapechar=\#,
    framexleftmargin=1em,
    framexrightmargin=1ex,
    keepspaces=true,
    keywordstyle=\color{dkblue},
    mathescape,
    numbers=none,
    numberblanklines=false,
    numbersep=1.25em,
    numberstyle=\relscale{0.8}\color{gray}\ttfamily,
    showstringspaces=true,
    stepnumber=1,
    xleftmargin=1em
}

\lstnewenvironment{code}[1][]
  {\small\lstset{language=custom-haskell,#1}}
  {}

\lstnewenvironment{expression}[1][]
  {\small\lstset{#1}}
  {}

\newcommand\bigforall{\mbox{\Large $\mathsurround0pt\forall$}}

\newcommand{\esummary}{e-summary\xspace}
\newcommand{\esummaries}{e-summaries\xspace}

\def\secref#1{\ref{sec:#1}}

\newtheorem{theorem}{Theorem}[section]

\newtheorem{lemma}[theorem]{Lemma}

\title{Hashing Modulo Alpha-Equivalence}

\author{Krzysztof Maziarz}
\affiliation{
\institution{Microsoft Research}
  \city{Cambridge}
  \country{UK}
}
\email{krmaziar@microsoft.com}

\author{Tom Ellis}
\affiliation{
\institution{Microsoft Research}
  \city{Cambridge}
  \country{UK}
}
\email{toelli@microsoft.com}

\author{Alan Lawrence}
\affiliation{
\institution{Microsoft Research}
  \city{Cambridge}
  \country{UK}
}
\email{allawr@microsoft.com}

\author{Andrew Fitzgibbon}
\affiliation{
\institution{Microsoft Research}
  \city{Cambridge}
  \country{UK}
}
\email{awf@microsoft.com}

\author{Simon Peyton Jones}
\affiliation{
\institution{Microsoft Research}
  \city{Cambridge}
  \country{UK}
}
\email{simonpj@microsoft.com}

\begin{document}

\begin{abstract}
In many applications one wants to identify identical subtrees of a program syntax tree.
This identification should ideally be robust to alpha-renaming of the program, but
no existing technique has been shown to achieve this with good efficiency (better than $\mathcal{O}(n^2)$ in expression size).
We present a new, asymptotically efficient way to hash modulo alpha-equivalence.
A key insight of our method is to use a weak (commutative) hash combiner at exactly one point in the construction, which admits an algorithm with $\mathcal{O}(n (\log n)^2)$ time complexity.
We prove that the use of the commutative combiner nevertheless yields a strong hash with low collision probability.
Numerical benchmarks attest to the asymptotic behaviour of the method.
\end{abstract}



\maketitle

\section{Introduction}
This paper addresses the problem of hashing abstract syntax trees while respecting alpha equivalence. This is a generic problem, with applications in many areas of programming language implementation, for example common subexpression elimination (CSE), hashing for structure sharing, or as part of pre-processing for machine learning.

Taking CSE as an example, consider the following program fragment
\begin{code}
(a + (v+7)) * (v+7)
\end{code}
A standard CSE transformation can rewrite this to
\begin{code}
let w = v+7 in (a + w) * w
\end{code}
which can be computed more efficiently.  However, CSE is not entirely straightforward. Consider
\begin{code}
(a + (let x = exp(z) in x+7)) *
     (let y = exp(z) in y+7)
\end{code}
We might hope that CSE would spot that the two let-bound terms are $\alpha$-equivalent, and transform to 
\begin{code}
let w = (let x = exp(z) in x+7) in (a + w) * w
\end{code}
We would like to similarly spot the equivalence of the two lambda terms in
\begin{code}
foo (\x.x+7) (\y.y+7)
\end{code}
and transform to 
\begin{code}
let h = \x.x+7 in foo h h
\end{code}
So, we want to find {\em all pairs}, or more precisely all equivalence classes, of {\em $\alpha$-equivalent subexpressions} of a given program.  Since the program may be large, we would like to generate the $\alpha$-equivalence classes in reasonable time. If there were a hash function invariant under $\alpha$-renaming that could be computed for every node in a single pass over the tree, the equivalence classes could be generated in the cost of a single sort.
Somewhat surprisingly, the CSE literature barely mentions the challenge of hashing modulo $\alpha$-equivalence, nor does the wider literature on hashing of program fragments.  (One might wonder whether switching to de Bruijn indexing would solve the problem,
but it does not, as we show in Section~\secref{deBruijn}.)

Another challenge is that in typical compilers the program is subjected to thousands of rewrites, each of which transforms the program locally.  
Ideally, we would like an \emph{incremental} hashing algorithm, so that we can continuously monitor sharing, for example for register pressure sensitive optimization algorithms.

In this paper we address these challenges, making the following contributions:
\begin{itemize}
    \item We present an algorithm that identifies all equivalence classes of subexpressions of an expression, respecting $\alpha$-equivalence (Section~\secref{key-ideas}).
    The algorithm is developed in two steps.  The first defines an \emph{e-summary} at each node; this step is invertible, allowing the original expression (modulo~$\alpha$) to be reconstructed  (Section~\secref{step1}).
    In the second step we develop
    a more efficient representation for \esummaries, optimized for the task of producing a hash code for the \esummary (Section~\secref{step2}).  This two-step approach makes the correctness argument easy (Section~\secref{outline}).
    \item We show that the algorithm runs in sub-quadratic time (Section~\secref{complexity}), and is compositional, so that it can readily be made incremental (Section~\secref{incremental}).
    \item A key step in making the algorithm efficient is to use a weak hash combiner (exclusive-or) when computing the hash of a finite map (Section~\secref{fast-hash-variable-map}).  
    At first glance, that weak combiner threatens the good properties of hashing.
    However, we compute
    the theoretical collision probability for our hash function, showing it can be upper-bounded, with the bound decreasing exponentially with the size of the hashing space (Section~\secref{proof-hash-is-strong}).
    \item 
    Our proof also lays down several lemmas about compositional hashing functions, which we believe will prove useful for analyses beyond the one done in this work.
\end{itemize}
We empirically evaluate the asymptotic behaviour of our approach in Section~\secref{empirical}, and discuss related work in Section~\secref{related}.

A Haskell implementation of our hashing algorithm together with all the benchmarks and baselines can be found at \url{https://github.com/microsoft/hash-modulo-alpha}.

\section{The Problem We Address}
\label{sec:problem}

Many algorithms were designed to analyse or transform programs.
These applications range from classical tools such as compilers and static analysis methods, to understanding and generating code using deep learning~\cite{balog2017deepcoder,allamanis2017learning,brockschmidt2018generative}.  The code being analysed or transformed is generally represented by an Abstract Syntax Tree (AST), which represents computational expressions using a tree structure.  Subtrees of such an AST --- referred to as \emph{subexpressions} --- are useful, because they often correspond to semantically meaningful parts of the program, such as functions. 

Many applications need to quickly identify all equivalent subexpressions in an AST.  Examples include \emph{common subexpression elimination} (CSE), as mentioned above; \emph{structure sharing} to save memory, by representing all occurrences of the same subexpression by a pointer to a single shared tree; or \emph{pre-processing for machine learning}, where subexpression equivalence can be used as an additional feature, for example by turning an AST into a graph with equality links.

\subsection{What Does ``Equivalent'' Mean?}
Downstream tasks may differ in what subexpressions they consider ``equivalent''. For example, here are four candidates:
\begin{itemize}
    \item \emph{Syntactic equivalence} means that two subexpressions are equivalent if they are identical trees; the same shape, with the same nodes, and the same variable names.
    \item \emph{$\alpha$-equivalence} is like syntactic equivalence but is insensitive to renaming of bound variables. For example, 
    the expression \lstinline|(\x.x+y)| is equivalent to \lstinline|(\p.p+y)| (by $\alpha$-renaming the lambda-bound variable), but not equivalent to \lstinline|(\q.q+z)|, because the free variables differ.
    \item \emph{Graph equivalence} goes beyond $\alpha$-equivalence by treating a \lstinline|let| expression as a mere textual description of a graph.  So \lstinline|(let x=e1 in let y=e2 in x+y)| is equivalent to \lstinline|(let y=e2 in let x=e1 in x+y)|, and to \lstinline|(e1+e2)|, because all three describe the same underlying graph.
    \item \emph{Semantic equivalence} says that two subexpressions are equivalent if they evaluate to the same value, regardless of the values of their free variables.  For example \lstinline|(3+x+4)| is equivalent to \lstinline|(x+7)| and \lstinline|(7+x)| among many others.
\end{itemize}
The difficulty of deciding equivalence ranges from trivial (syntactic equivalence) to undecidable (semantic equivalence).
In this paper \emph{we focus on $\alpha$-equivalence}.
We specifically do not want to go as far as graph equivalence, because let-expressions express operational choices about object lifetimes and evaluation order, and so graph equivalence is too strong for the downstream tasks that we are interested in. Of course, graph equivalence might be just right for other applications, and it would be interesting to adapt the ideas presented here, but we leave exploring that to future work.

\subsection{Baseline: Purely Syntactic Equivalence} \label{sec:syntactic-equivalence}

Purely syntactic equivalence is easy, and perfect for structure sharing, but not for much else.  For example, it is inadequate for CSE, and other tasks, via two primary failure modes: false negatives and false positives.

\begin{itemize}
\item \emph{False negatives: sensitivity to arbitrary variable names}.
Consider this expression:
\begin{code}
map (\y.y+1) (map (\x.x+1) vs)
\end{code}
The two lambda-expressions are not syntactically identical, but they are $\alpha$-equivalent, 
and perform the same computation in the same way. Similarly, consider
\begin{code}
foo (let bar = x+1 in bar*y)
    (let pub = x+1 in pub*y)
\end{code}
Here we would like to CSE the two arguments to \lstinline{foo}, even though they use different binders internally.

    \item \emph{False positives: name overloading}.
Consider the syntactically repeated subexpression \lstinline{x+2} in this example:
\def\subex#1{\underline{#1}}
\begin{code}
foo (let x=bar in #\subex{x+2}#) (let x=pub in #\subex{x+2}#)
\end{code}
The two subexpressions \lstinline{x+2} are unrelated, but they are syntactically identical. If the goal is structure sharing this is fine; indeed we might \emph{want} to share the two \lstinline{x+2} subexpressions, to save memory.  However, sharing the two would be wrong for tasks similar to CSE. For example, it would be clearly wrong to transform the above expression into
\begin{code}
let tmp = x+2
 in foo (let x=bar in tmp) (let x=pub in tmp)
\end{code}
\end{itemize}
\noindent
The second problem can readily be addressed, by preprocessing the expression so that every binding site binds a distinct variable name.  This step takes time linear in the expression size $n$; or, more precisely, if we take account of the $\mathcal{O}(\log n)$ time to look up a bound variable in the environment,  $\mathcal{O}(n \log n)$. We assume this preprocessing in all algorithms below.
\label{sec:unique-names}

\subsection{Hashing For Syntactic Equivalence} \label{sec:syntactic-equiv}

A standard approach to determining subexpression equivalence is to use some form of \emph{hashing}. We compute a fixed-size \emph{hash code} for each node in the tree, and use these hash codes to insert every node into a \emph{hash table}.

Hashing gives a simple, direct, and compositional implementation for syntactic equivalence: the hash for a node is computed by hashing the node constructor with the hashes of the children.  When used for structure sharing (to save memory), this is often called \emph{hash-consing}. When you are just about to allocate a new node, first compute its hash code, and then look that up in the hash table to see if that node already exists. If so, use it; if not, allocate it and add it to the hash table.  In effect, we simply memoise the node constructor functions.

But this simple approach fails for $\alpha$-equivalence, because \lstinline|(\x.x+1)| and \lstinline|(\y.y+1)| have different hash codes.  How can we fix up the hashing approach to account for $\alpha$-equivalence?

\subsection{De Bruijn Indexing}
\label{sec:deBruijn}
One well-known way to become insensitive to $\alpha$-renaming is to use a nameless representation, typically \emph{de Bruijn indexing}.  Lambdas have no binder; and each occurrence of a bound variable is replaced by a number that counts how many intervening lambdas separate that occurrence from its binding lambda.  For example, the expression \lstinline|(\x.\y.x+y*7)| is represented in de Bruijn form by \verb|(\.\.%1+%0*7)|, where we use \verb|%|$i$ to represent a variable occurrence with de Bruijn index~$i$.

After switching to de Bruijn indexing, can we use vanilla hashing to determine equivalence,
and thus solve the $\alpha$-equivalence challenge? Sadly, no: using de Bruijn indexing remains vulnerable to both false positives and false negatives:
\begin{itemize}
    \item \emph{False negatives}.  Consider
\begin{code}
\t. foo (\x.x+t) (\y.\x.x+t)
\end{code}
The two subexpressions \lstinline{(\x.x+t)} are certainly equivalent, and we could profitably transform the expression to
\begin{code}
\t. let h = \x.x+t in foo h (\y.h)
\end{code}
But with de Bruijn indexing, the two expressions look different:
\begin{code}
\. foo (\.%0+%1) (\.\.%0+%2)
\end{code}
Notice how the two occurrences of \lstinline{t} have become \verb|%1| and \verb|%2| respectively.

\item \emph{False positives}.  Consider
\begin{code}
\t. foo (\x.t*(x+1)) (\y.\x.y*(x+1))
\end{code}
With de Bruijn indexing this looks like
\begin{code}
\. foo (\.%1*(%0+1)) (\.\.%1*(%0+1))
\end{code}
This has two occurrences of \verb|(\.%1*(%0+1))|, which might be great for structure sharing, but is wrong for CSE.

Note that, unlike with simple syntactic equivalence,
these false positives cannot be eliminated by giving every binder a unique name -- with
de Bruijn there \emph{are} no binders!
\end{itemize}
\noindent
Moreover, using de Bruijn indexing as the internal representation of an expression in a compiler
incurs serious costs of its own, because terms need to be repeatedly traversed as they are substituted under other binders, to
adjust their de Bruijn indices.
We know of no systematic, quantitative comparison of the engineering tradeoff between de Bruijn and named representations in a substantial application (e.g. a compiler), perhaps because the choice has such pervasive effects that implementors are typically forced to make one choice or the other, and stick to it. A decent attempt was made in \cite{murphy:tilt}, with the conclusion that the costs of a de Bruijn representation exceed the benefits.

\subsection{Locally Nameless} \label{sec:locally-nameless}

The false negatives and false positives of de Bruijn indexing are a serious problem for CSE-like purposes.
However, they can be avoided by using the ``locally nameless'' representation~\cite{mcbride04,weirich,chargueraud12}.
The idea is simple: the hash of an expression is defined to be the hash of
the de-Bruijn-ised representation of \emph{the subexpression taken in isolation}.  For example, the hash of \lstinline{(f x (\y.x+y))} would be the hash of
\verb|(f x (\.x+%0))|.  In this expression the free variables \lstinline{f} and \lstinline{x} remain unchanged, but the locally-bound variable \lstinline{y} has been de-Bruijn-ised.  The expression to be hashed has a mixture of de Bruijn indices and named free variables.

The hash of an application \lstinline{(e1 e2)} can be obtained by combining the hash of \lstinline{e1} and \lstinline{e2};
but the hash of \lstinline{(\x.e)} \emph{cannot} be obtained from the hash of \lstinline{e}.
The hash of \lstinline{e} incorporates the hash of each occurrence of \lstinline{x} in \lstinline{e}; but the hash
of \lstinline{(\x.e)} must instead incorporate the hash of an appropriate de Bruijn index at each of those occurrences.
We cannot do this compositionally; instead, we must first de-Bruijn-ise \lstinline{x} in \lstinline{e}, and then take
the hash of \emph{that}.

This algorithm
correctly does hashing modulo $\alpha$-equi\-val\-ence, but it comes with a cost in asymptotic complexity: as we pass each lambda, we must re-hash the entire body.

In practice, the locally-nameless scheme works sufficiently well that it is used in Epigram \cite{mcbride04} and the LEAN theorem prover \cite{lean}.
However, it suffers from the same complexity issues as de Bruijn.   Other things being equal, we would prefer to avoid compiler technology that has asymptotic complexity holes fundamentally built in, as well as the other index-shuffling costs imposed by a de-Bruijn-based (including locally nameless) representation \cite{murphy:tilt}.  Our contribution is to show how to use a name-ful representation (avoiding index shuffling), and still
get compositional hashing with asymptotically-good complexity.

\section{The Key Ideas} \label{sec:goal}\label{sec:key-ideas}

\begin{figure*}[t]
\includegraphics[width=\linewidth]{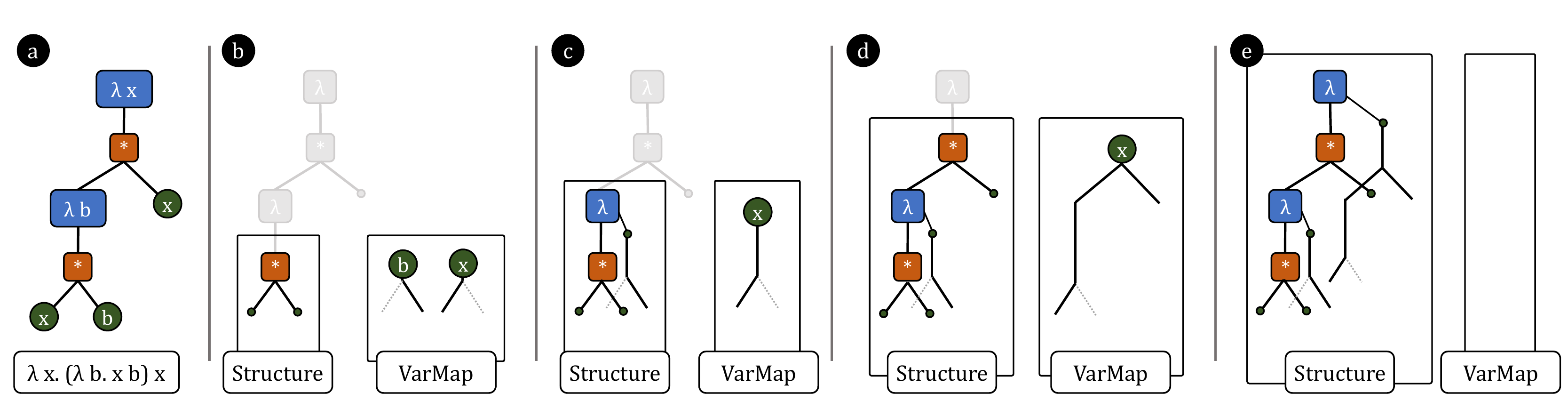}
\caption{ (a) Input expression, with names at \lstinline{Lam} and \lstinline{Var} nodes. (b-e) E-Summaries for subexpressions, names only in the VarMap. This diagram depicts an $\mathcal{O}(n^2)$ algorithm (Section~\secref{full-algo-quadratic}); then the ``smaller subtree'' (Section~\secref{smaller-tree}) and ``xor'' (Section~\secref{xor}) modifications reduce complexity for hash computation to $\mathcal{O}(n (\log n)^2)$. }
\label{fig:esummaries}
\end{figure*}

In this section we describe the key ideas of our approach.  Before doing so, it is helpful to articulate our goal more precisely.

\noindent
{\bf Goal.}  Given an expression $e$, in which every binding site binds a distinct variable name, identify all equivalence classes of subexpressions of $e$, where two subexpressions are equivalent if and only if they are $\alpha$-equivalent. We wish to achieve this goal in a way that is:

\begin{itemize}
    \item {\bf Compositional}. If we have already done the computation for \lstinline{e1} and \lstinline{e2}, computing the result for \lstinline{(e1 e2)} should be done by combining the results from children. In particular, context-dependent computation is not allowed.
    \item {\bf Efficient.}  We consider this to mean that finding all equivalent subexpressions should be sub-quadratic in the size of the expression.
\end{itemize}
\noindent
Compositionality helps with efficiency, because combining two smaller expressions \lstinline{e1} and \lstinline{e2} into a bigger one involves only combining the results of processing these subexpressions.  But, crucially for our applications, compositionality also allows hashing to be \emph{incremental}: if we have already performed the equality-discovery task for a large expression, and we make a small rewrite in that expression, we can efficiently recompute the results by examining only parts of the expression that have changed.  We give an analysis of incrementality in Section~\secref{incremental}.

As to complexity, we consider an algorithm that is quadratic in expression size to be too expensive. Linear (constant work at each node) would be ideal; in this paper we achieve log-linear (generally $\mathcal{O}(n(\log n)^k)$, with here $k=2$), which we consider acceptable.

\subsection{The Challenge of Compositionality} \label{sec:compositional-challenge}

Given an application \lstinline{(e1 e2)}, a compositional algorithm will somehow process \lstinline{e1} and \lstinline{e2} separately, and combine those results to get the result for \lstinline{(e1 e2)}.  Let us call ``the result of processing an expression'' the \emph{\esummary} for that expression.  You could imagine attaching the \esummary for \lstinline{e} to the tree node for \lstinline{e}, and computing the \esummary for \lstinline{(e1 e2)} from the e-summaries for \lstinline{e1} and \lstinline{e2}.  The idea is that an \esummary precisely identifies the equivalence class: two subexpressions are $\alpha$-equivalent iff their \esummaries are equal.

The difficulty with this approach is that the expressions \lstinline{(x+2)} and \lstinline{(y+2)} are different, and must have \emph{different} \esummaries, but \lstinline{(\x.x+2)} and \lstinline{(\y.y+2)} are $\alpha$-equivalent, and must have the \emph{same} \esummary.  So an \esummary cannot be a simple numeric hash code, because there is no way to get the hash-code for, say, \verb|%0+2| from the hash-code of \verb|x+2|.

Therefore, we need a richer \esummary.  At one extreme, an expression could of course be its own \esummary! That would be fast to compute (a no-op), and compositional, but asking if two e-summaries are equal would require an $\alpha$-respecting equality comparison between the two summaries. The task of finding all equivalence classes would be absurdly expensive, requiring an $\alpha$-respecting equality comparison between every pair of expressions.

We seek something in between the two. We will define an \esummary that is not fixed-size (like a hash code), but from which we can rapidly compute a hash code.

\subsection{Overview of Our Approach} \label{sec:outline}

A general overview of our approach is as follows. 

\paragraph{Step 1 (Section~\secref{step1}).} 
We define a particular \emph{\esummary}, with the following properties:
\begin{itemize}
    \item The \esummary for an expression can be computed in a compositional way (Section~\secref{full-algo-quadratic}).
    The cost of computing it for all subexpressions is quadratic in expression size, a problem we fix later.
    \item The \esummary for $e$ can be converted back to an expression $e^\prime$ which is $\alpha$-equivalent to $e$ (Section~\secref{rebuild}).  That is, 
    \esummaries lose no information, except the names of the bound variables.
    \item Two \esummaries are equal if and only if the expressions from whence they came are $\alpha$-equivalent.
\end{itemize}
\noindent
\paragraph{Step 2 (Section~\secref{step2}).}  At this stage, it may appear that not much has been gained.  An expression $e$ and its \esummary are inter-convertible, and comparing \esummaries is not much faster than comparing the corresponding expressions. But, \esummaries enjoy a crucial extra property: \emph{unlike expressions, we can easily represent an \esummary in a hashed form that is much more compact and enjoys $\mathcal{O}(1)$ comparison time}.

The two-step approach makes the algorithm easier to reason about.
Step 1 loses no information, and hence cannot give rise to false positives.
Step 2 is just regular hashing and, like any hash, can suffer from collisions and hence false positives;
but in Section~\secref{proof-hash-is-strong} we show that the probability of collision remains inversely exponential in the number of bits in the hash code.

\section{Step 1: A Compositional E-Summary}  \label{sec:step1}

In this section we give the details of our \esummary. Although it does not yet give a way to efficiently find equal subexpressions, it lays the groundwork for Step 2 discussed in Section~\secref{step2}.
Splitting the development in two in this way allows much easier reasoning about correctness.

For the sake of concreteness, we  use  Haskell to express our algorithm, its data types
and its functions.

\subsection{Preliminaries} \label{sec:preliminaries}

First, we need a datatype for representing the expression:
\begin{code}
type Name = String

data Expression = Var Name
                | Lam Name Expression
                | App Expression Expression 
\end{code}
\noindent
For simplicity, we use \texttt{String} for variable names and assume they can be compared in constant time; a practical implementation should replace the \texttt{String} names with unique identifiers that support constant-time comparison. As specified in Section~\secref{goal}, we assume that the variable names are all unique. This requirement is easy to satisfy by renaming the variables during preprocessing.

This language is very minimal, but it is enough to demonstrate the $\alpha$-equivalence challenge, and it can readily be extended to handle richer binding constructs (let, case, etc.), as well as
constants, infix function application, and so on.

\subsection{The Basic E-Summary} \label{sec:rebuild}

Recall from Section~\secref{compositional-challenge} that adding a lambda at the root of an expression
must transform distinct \esummaries (for, say, \lstinline|x+2| and \lstinline|y+2|) into the same one (for \lstinline|\x.x+2|).  To account for this,
we define the \esummary of an expression $e$ to be a pair of:
\begin{itemize}
 \item The \emph{structure}, or shape, of $e$ (Section~\secref{structure}). The structure of $e$ completely describes $e$ apart from its free variables (imagine every free variable being replaced by \lstinline{<hole>}, 
 so \lstinline{(add x y)} has same structure as \lstinline{(add x x)}).
    \item The \emph{free-variable map} of $e$ (Section~\secref{free-variables-map}). The variable map of $e$ is a list of $e$'s free variables, each with a \emph{tree of positions} in $e$ where it occurs (we define positions in Section~\secref{positions}).
\end{itemize}
\noindent
Therefore, an \esummary is a pair of a structure and a free-variable map:
\begin{code}
  data ESummary = ESummary Structure VarMap
\end{code}
We will elaborate each of these types in the following sections. In Figure~\ref{fig:esummaries}, we show how \esummaries are built up for an example expression.
Our basic algorithm has the following signature:
\begin{code}
  summariseExpr :: Expression -> ESummary
\end{code}
\noindent
The function \lstinline{summariseExpr} converts an \lstinline{Expression} into its \lstinline{ESummary}. 
A key correctness property is that we can reconstruct an \lstinline{Expression} from an \lstinline{ESummary}, up to the names of bound variables. That is, it is possible to implement
\begin{code}
rebuild :: ESummary -> Expression
\end{code}
so that \lstinline{rebuild (summariseExpr e)} is $\alpha$-equivalent to \lstinline{e}; in other words, an \lstinline{ESummary} loses no necessary information. Indeed, others have suggested using a representation in which a lambda contains a list of the occurrences of its bound variable as the \emph{primary} representation of lambda terms \cite{Abel_2011,sato_2013,McBride_2018}.

\subsection{Expression Structure}\label{sec:structure}

The \emph{structure} of an expression $e$ expresses the shape of $e$, \emph{ignoring the identity of its free variables}.
\begin{code}
  data Structure
    = SVar                      -- Anonymous
    | SLam (Maybe PosTree) Structure
    | SApp Structure Structure 
\end{code}
Variables are replaced by an anonymous \lstinline{SVar}.  A lambda does not \emph{name} its bound variable; instead, it lists the positions at which that bound variable occurs in its body.  Of course an actual list, of the form \lstinline|{L,LLRL,RRL}| would be alarmingly inefficient, instead {\em position trees} (of type \lstinline{PosTree}) are used, as described in Section~\secref{positions}.

We will build values of type \lstinline|Structure| using ``smart constructors''.
\begin{code}
  mkSVar :: Structure
  mkSLam :: Maybe PosTree -> Structure -> Structure
  mkSApp :: Structure -> Structure -> Structure
\end{code}
You can think of these as simply renamings of the underlying data constructors (e.g. \lstinline|mkSVar = SVar|), but
in Section~\secref{hash-structure} we will  exploit the flexiblity of being able to redefine \lstinline|mkSVar|.

\subsection{Free Variables Map}\label{sec:free-variables-map}

The free-variable map of an expression maps each free variable of $e$ to the positions at which that variable occurs. It supports the following operations:

\begin{code}
emptyVM      :: VarMap
singletonVM  :: Name -> PosTree -> VarMap
extendVM     :: Name -> PosTree -> VarMap -> VarMap
removeFromVM :: Name -> VarMap 
             -> (VarMap, Maybe PosTree)
    -- Removes one item from the map, returning what
    -- the variable mapped to, or Nothing if it
    -- was not in the map
toListVM :: VarMap -> [(Name, PosTree)]
\end{code}
\noindent
One possible implementation of \lstinline|VarMap| is to use Haskell's \lstinline|Data.Map| library:
\begin{code}
  type VarMap = Map Name PosTree
\end{code}
We will introduce a few more operations on \lstinline{VarMap} as we go along; all can be implemented straightforwardly using standard libraries.

\subsection{Position Trees}\label{sec:positions}

A value of type \lstinline{PosTree} identifies a set of one or more \lstinline{SVar} nodes inside a \lstinline|Structure|.  A \lstinline{PosTree} is a skeleton tree, with the same structure as the expression,  reaching only the leaves of the expression that are occurrences of one particular variable:
\begin{code}
data PosTree
  = PTHere
  | PTLeftOnly PosTree
  | PTRightOnly PosTree
  | PTBoth PosTree PosTree
\end{code}
So the occurrences of variable \lstinline{"x"} in
\begin{code}
App (App (Var "f") (Var "x")) (Var "x")
\end{code}
are described by the position tree
\begin{code}
PTBoth (PTRightOnly PTHere) PTHere
\end{code}
In a \lstinline{Structure}, an \lstinline{SLam} node contains a position tree that describes all the occurrences of that variable.  A position tree always represents \emph{one or more} occurrences, so
an \lstinline{SLam} node actually contains a \lstinline{(Maybe PosTree)}, with \lstinline{Nothing} indicating that the bound variable does not occur at all in the body of the lambda.

As with \lstinline{Structure}, we use ``smart constructors'' (\lstinline{mkPTHere}, \lstinline{mkPTBoth}, etc.) so that we can  give these ``constructors'' extra behaviour in Section~\secref{hash-structure}.

\subsection{Full Algorithm}\label{sec:full-algo-quadratic}

After designing the auxiliary data structures, we may instantiate our algorithm as follows
\begin{code}
data ESummary = ESummary Structure VarMap

summariseExpr :: Expression -> ESummary
summariseExpr (Var v)
  = ESummary mkSVar (singletonVM v mkPTHere)

summariseExpr (Lam x e)
  = ESummary (mkSLam x_pos str_body) vm_e
  where
    ESummary str_body vm_body = summariseExpr e
    (vm_e, x_pos) = removeFromVM x vm_body

summariseExpr (App e1 e2)
  = ESummary (mkSApp str1 str2) (merge vm1 vm2)
  where
    ESummary str1 vm1 = summariseExpr e1
    ESummary str2 vm2 = summariseExpr e2
    merge = 
         mergeVM mkPTLeftOnly mkPTRightOnly mkPTBoth

mergeVM :: (PosTree -> PosTree)  -- Left only
        -> (PosTree -> PosTree)  -- Right only
        -> (PosTree -> PosTree -> PosTree)  -- Both
        -> VarMap -> VarMap -> VarMap
\end{code}
Most of the work is done in \lstinline{App} nodes, where we need to combine variable maps from the node's children. To that end, we use a new function \lstinline{mergeVM}, which combines the position trees from the children's maps. The three argument functions to \lstinline{mergeVM} say what to do if only the left map has the variable in its domain, only the right map does, or both.  In the call to \lstinline{mergeVM} we simply use the constructors from \lstinline{PosTree} for these three cases.

The time complexity of this version of our algorithm is quadratic,
because at each \lstinline{App} node the \lstinline{mergeVM} operator must touch every element of the domain of the mapping, taking time proportional to the number of free variables of the expression. 
In Section~\secref{smaller-tree}, we discuss the key optimization needed to bring the complexity down to log-linear. However, we first prove that the conversion to \lstinline{ESummary} is reversible by designing the \lstinline{rebuild} function.

\subsection{Rebuilding}

The \lstinline{rebuild} function (Section~\secref{rebuild}) is easy to write
\begin{code}
rebuild :: ESummary -> Expression
rebuild (ESummary SVar vm) =
  Var (findSingletonVM vm)
rebuild (ESummary (SLam p s) vm) =
  Lam x (rebuild (ESummary s (extendVM x p vm)))
    where x = ... -- fresh variable name
rebuild (ESummary (SApp s1 s2) vm) =
  App (rebuild (ESummary s1 vm1))
      (rebuild (ESummary s2 vm2))
    where m1 = mapMaybeVM pickL vm
          m2 = mapMaybeVM pickR vm
    
pickL :: PosTree -> Maybe PosTree
pickL (PTLeftOnly p)  = Just p
pickL (PTBoth pl _)   = Just pl
pickL _               = Nothing

findSingletonVM :: VarMap -> Name
  -- The map should be a singleton map;
  -- return its unique key
mapMaybeVM :: (PosTree -> Maybe PosTree)
           -> VarMap -> VarMap
    -- Apply the function to every element of
    -- the map; delete if the function returns Nothing
\end{code}
\noindent
In this function we use two new functions over \lstinline{VarMap}:
\begin{itemize}
    \item \lstinline{findSingletonVM} expects its argument to be a singleton map, and returns
    the unique \lstinline{Name} from its domain (which should be mapped to \lstinline{PTHere}).
    The function fails if the map is not a singleton, but that should not happen if the \lstinline{ESummary} is well-formed.
    \item \lstinline{mapMaybeVM} applies a function to every element of the domain
    of the map; if the function returns \lstinline{Nothing} that element is deleted.
\end{itemize}
\noindent
In the \lstinline{SLam} case we have to invent a fresh variable name, since the original name is not 
recorded, and hence the returned expression is only $\alpha$-equivalent to the original,
not identical\footnote{
As an alternative, it would
be easy to record that name in the \lstinline{Structure}, to recover \emph{precisely}
the original expression, rather than just an $\alpha$-equivalent one. If we did so, this 
name should not participate in calculation of the hash values described in Section~\secref{step2}.}.

Why do we go to the trouble of defining \lstinline|rebuild|, which is not even part of the original problem specification?  \emph{We define \lstinline|rebuild| because its existence guarantees that our \esummary is not information-losing}, and that in turn guarantees that the hash-code for an \esummary will have few collisions (assuming it is a strong hash). This is important: for example,
consider a degenerate, information-losing \esummary that recorded only the \emph{size} of the tree; it would be fast and compositional, but its information loss would lead to rampant false positives.

In the next section we will optimize the \esummary to improve the complexity of \lstinline|summariseExpr|,
using the \lstinline|rebuild| function to drive our decisions about what information the \esummary needs to record.

\subsection{Using the Smaller Subtree} \label{sec:smaller-tree}

When processing an \lstinline{App} node, the algorithm from Section~\secref{full-algo-quadratic} uses \lstinline{mergeVM} which transforms every element of its range, thereby taking time proportional
to the number of free variables of the expression.
For very unbalanced trees this might be quadratically expensive. In this section we modify the algorithm so that it only transforms the \emph{smaller} map, leaving the other unchanged. The more unbalanced the tree, the less traversal we do; the worst case becomes a balanced tree, and that has only $\mathcal{O}(n \log n)$ complexity.

First, we augment the \lstinline{Structure} datatype with a \lstinline|Bool| flag in \lstinline|SApp| that records
which child has more free variables:
\begin{code}
  data Structure
    = SVar
    | SLam (Maybe PosTree) Structure
    | SApp Bool Structure Structure
           -- True  if the left  expr has more free vars
           -- False if the right expr has more free vars
\end{code}
\noindent
Now, the key \lstinline{App} case of \lstinline{summariseExpr} becomes

\begin{code}
summariseExpr (App e1 e2) = ESummary str vm
 where
  ESummary str1 vm1 = summariseExpr e1
  ESummary str2 vm1 = summariseExpr e2
  str = mkSApp left_bigger str1 str2
  tag = structureTag str
  vm  = foldr add_kv big_vm (toListVM small_vm)
  left_bigger = vm1 `isBiggerThanVM` vm2
  (big_vm, small_vm) = if left_bigger
                       then (vm1, vm2)
                       else (vm2, vm1)
  add_kv :: (Name, PosTree) -> VarMap -> VarMap
  add_kv (v, p) vm 
    = alterVM (\mp -> mkPTJoin tag mp p) v vm
  
alterVM :: (Maybe PosTree -> PosTree)
        -> Name -> VarMap -> VarMap
  -- Alter the value to which the key is mapped
\end{code}
As you can see from the definition of \lstinline{vm}, we convert the \emph{smaller} map
to a list of key-value pairs using \lstinline{toListVM}, and add them one at a time to the larger map using \lstinline{add_kv}.
The new function \lstinline{alterVM} alters the mapping at one key; the argument function allows the caller
to behave differently depending on whether or not the key was in the map beforehand.
But what is this mysterious \lstinline{tag} and the new \lstinline{mkPTJoin} operation on position trees?
\footnote{Readers who feel there must be a more mathematically elegant way to do this might enjoy Appendix~\ref{appendix:number-theory}, but the way described here is simple and fast.}

First, \lstinline{structureTag} extracts from a \lstinline{Structure} some kind of ``tag'' (an integer, say)
\begin{code}
type StructureTag = Int
structureTag :: Structure -> StructureTag
\end{code}
This function must satisfy one simple property: \emph{a structure must have a different tag to the tag of any of its sub-structures}. 
We abstract away the exact implementation of \lstinline{structureTag}, 
but one simple possibility is to have it return the depth of the \lstinline{Structure}, 
which can be computed and stored at the point when a \lstinline{Structure} is constructed.

Next, here is the new definition of \lstinline{PosTree}:
\begin{code}
data PosTree
  = PTHere
  | PTJoin StructureTag
           (Maybe PosTree)  -- Child from bigger map
           PosTree          -- Child from smaller map
\end{code}
As you can see from \lstinline{add_kv}, we make a tagged \lstinline{PTJoin} for every variable in the \emph{smaller} map,
but variables that appear only in the larger map are left untouched.
The tag allows \lstinline{rebuild} to invert this combining operation, in a unique way, determined by whether or not each item is tagged with the tag for this structure:
\begin{code}
rebuild (ESummary str@(SApp left_bigger s1 s2) vm) 
  = App (rebuild (ESummary s1 vm1))
        (rebuild (ESummary s2 vm2))
  where
    tag = structureTag str
    small_m = mapMaybeVM upd_small vm
    big_m   = mapMaybeVM upd_big   vm
    (vm1, vm2) = if left_bigger
                 then (big_m, small_m)
                 else (small_m, big_m)
                
    upd_small :: PosTree -> Maybe PosTree
    upd_small (PTJoin ptag mpt pt)
       | ptag == tag = Just pt
    upd_small _ = Nothing
    
    upd_big :: PosTree -> Maybe PosTree
    upd_big (PTJoin ptag mpt pt)
       | ptag == tag = mpt
    upd_big pt = Just pt
\end{code}
The \lstinline{left_bigger} flag in \lstinline{SApp} tells whether the bigger map came from the left or right argument. The tag in \lstinline{PTJoin} tells whether it belongs to the \lstinline{SApp} under consideration, or belongs to one deeper in the structure.

Note that in this version of \lstinline{summariseExpr}, the amount of work done in an \lstinline{App} node is proportional to the size of the \emph{smaller} variable map from the node's children.

\section{Step 2: Hashing an E-Summary}\label{sec:step2}

To obtain an integer hash value that can be used for downstream tasks, we will use the following functions
\begin{code}
hashStructure :: Structure -> HashCode
hashVM        :: VarMap    -> HashCode
hashESummary  :: ESummary  -> HashCode
\end{code}
\noindent
Our aim is for all of these functions to work in $\mathcal{O}(1)$ time.
Conversion from an \esummary to \lstinline{HashCode} is information-losing, and non-invertible.

To implement \lstinline{hashESummary}, we may simply do
\begin{code}
hashESummary (ESummary str map) =
  hash (hashStructure str, hashVM map)
\end{code}
We deal with \lstinline|hashStructure| in Section~\secref{hash-structure}, and \lstinline|hashVM| in Section~\secref{fast-hash-variable-map}.

\subsection{Hashing Structures}\label{sec:hash-structure}

\lstinline{hashStructure} can easily be implemented by computing the hash at construction time, and storing it in the \lstinline{Structure} object itself.
That would be enough to achieve the complexity bound we desire.
But there is an even more attractive possibility: since \lstinline{hashStructure} is the \emph{only} function we will need for structures,
we can \emph{represent a structure simply by its hash code}, dispensing entirely with the tree, thus
\begin{code}
type Structure = HashCode

-- "Constructors" combine hash values
mkSVar :: Structure
mkSLam :: Maybe PosTree -> Structure -> Structure
mkSApp :: Structure -> Structure -> Structure

hashStructure :: Structure -> HashCode
hashStructure s = s
\end{code}
The ``constructors'' of the tree are implemented by $\mathcal{O}(1)$ hash combiners, and \lstinline{hashStructure} becomes the identity function.
We can apply precisely the same reasoning to \lstinline|PosTree|, and represent a value of type \lstinline{PosTree} by its \lstinline{HashCode}.

Of course, identifying each \lstinline{Structure} and \lstinline|PosTree| with its \lstinline{HashCode} has a much lower constant factor than representing structures and positions as trees: instead, we only manipulate hash codes.
In exchange, we will no longer be able to write \lstinline{rebuild}. However, recall that \lstinline|rebuild| is not used in the final implementation; its only purpose is that its existence shows the correctness of the algorithm.  By thinking first in terms of the non-information-losing data structure, and then thinking of efficient representations of those structures, we can get \emph{both} an easy correctness argument \emph{and} an efficient implementation.

\subsection{Hashing Variable Maps}\label{sec:fast-hash-variable-map}

Hashing variable maps is a little more tricky.
It would be prohibitively (indeed asymptotically) slow to 
compute the hash of the variable map afresh at each node. Instead, as for structures,
we would like to compute the hash of a node's variable map using the hashes of its children.
Doing so is far from trivial.  We might try to pair a map with its hash, thus:
\begin{code}
data VarMap = VM (Map Name PosTree) HashCode
\end{code}
\noindent 
and try to compute the hash for \lstinline{(f vm args)}, where \lstinline{f} is a function
that returns a new \lstinline{VarMap}, from the hash of \lstinline{vm} and \lstinline{f}'s other
arguments \lstinline{args}.
But consider \lstinline{removeVM}: how can we start with the hash of a map, and compute the hash of a map from which a particular entry has been removed?

\label{sec:xor}
Our key idea is this: \emph{we define the hash of a variable map as the XOR, written $\oplus$, of the hashes of its entries}, where an entry is a (variable, position-tree) pair $(v, p)$.  This definition has big advantages:
\begin{itemize}
    \item Since $\oplus$ is commutative and associative, it does not matter in which order we consider the entries.
    \item We can compute the hash of removing $(v, p)$ from a map $m$  by simply XORing $m$'s hash with the hash of $(v, p)$, since $(a \oplus b) \oplus a = b$
\end{itemize}
\noindent
More generally, we could use any operator $\oplus$ that is associative, commutative, and invertible.  
The trouble is that XOR is a cryptographically weak hash combiner, so using it to combine
hashes in this way looks suspicious---won't we get lots of unwanted collisions?
Fortunately, these fears are unfounded: we prove in Section~\secref{proof-hash-is-strong} that in our algorithm the use of XOR does not lead to excess hash collisions.

Computing hashes is now rather easy. The algorithm of Section~\secref{smaller-tree} needs only \lstinline{singletonVM}, \lstinline{alterVM}, and \lstinline{removeFromVM}:
\begin{code}
-- Arbitrary implementation - can be the builtin hash
entryHash :: Name -> PosTree -> HashCode
entryHash key pos = hash (key, pos)

singletonVM :: Name -> PosTree -> VarMap
singletonVM key pos = VM (Map.singleton key pos)
                         (entryHash key pos)

alterVM :: (Maybe PosTree -> PosTree)
        -> Name -> VarMap -> VarMap
alterVM f key (VM entries old_hash)
  | Just old_pt <- lookupVM entries key
  , let new_pt = f (Just old_pt)
  = VM (Map.insert key new_pt entries
       (old_hash $\oplus$ entryHash key old_pt
                 $\oplus$ entryHash key new_pt)
  | otherwise
  , let new_pt = f Nothing
  = VM (Map.insert key new_pt entries)
       (old_hash $\oplus$ entryHash key new_pt)

removeFromVM :: Name -> VarMap
             -> (VarMap, Maybe PosTree)
-- Deletes a Name from the VarMap, returning
-- its current PosTree, or Nothing if it was not in the map
removeFromVM key map@(VM entries old_hash)
  | Just pt <- Map.lookup key entries
  = (VM (key `Map.delete` entries)
        (old_hash $\oplus$ entryHash key pt), Just pt)
  | otherwise
  = (map, Nothing)
\end{code}

\section{Analysis}

In this section, we formally analyze the time complexity of our final algorithm, and then upper-bound the probability of obtaining incorrect results due to hash collisions. Throughout this section we use $|e|$ to denote the number of nodes in expression $e$.

\subsection{Time Complexity} \label{sec:complexity}

We start by bounding the amount of work done in \lstinline{App} nodes, and then derive the time complexity for the entire algorithm.  We present first a formal proof, and then an intuitive argument which may help the reader to see why the formal proof works.

\begin{lemma}\label{lemma:map-ops-app}
Let $e$ be an expression. The total number of \lstinline{alterVM} and \lstinline{removeFromVM} operations performed in \lstinline{App} nodes by \lstinline{summariseExpr} ran on $e$ is $\mathcal{O}(|e| \log |e|)$.
\end{lemma}

\begin{proof}
Denote the number of operations in question as $O_{App}(e)$. Let us define
\begin{equation*}
    T(n) = \max_{e : |e| \leq n} O_{App}(e)\text{.}
\end{equation*}
We will prove that $T(n)$ is $\mathcal{O}(n \log n)$, which will conclude the proof of the lemma.

First, consider a single \lstinline{App} node $v$ in $e$ with children $v_1$, $v_2$. The number of map operations performed when processing $v$ is $\mathcal{O}(\min(m_1, m_2))$, where $m_i = |map_i|$ is the size of the free variables map from~$v_i$. Since a free variables map only contains variables that are used in a given subtree, its size is bounded by the number of nodes in the subtree - i.e. $m_i \leq |v_i|$, and therefore $\min(m_1, m_2) \leq \min(|v_1|, |v_2|)$.

From the analysis above we get that for $n > 1$
\begin{equation}\label{eq:time}
    T(n) \leq \max_{1 \leq a < n} \left( T(a) + T(n - 1 - a) + C \cdot \min(a, n - 1 - a) \right)\text{,}
\end{equation}
where $a$ and $n - 1 - a$ correspond to $|v_1|$ and $|v_2|$, respectively, and $C$ is a constant resulting from the use of $\mathcal{O}$ notation. Due to symmetry, we can rewrite Equation~\ref{eq:time} as
\begin{equation}\label{eq:time-2}
    T(n) \leq \max_{1 \leq a \leq \frac{n - 1}{2}} \left( T(a) + T(n - 1 - a) + Ca \right)\text{.}
\end{equation}

Now, we will prove inductively that $T(n) \leq C n \log_2 n$. The base case of $n = 1$ holds, since $T(1) = 0$ as an expression consisting of a single node cannot have any \lstinline{App} nodes. 
Then
\def\eqnote#1{&&\hspace{-2em}\text{\textcolor{gray}{\small #1}}}
\begin{eqnarray*}
    T(n) &\leq &\max_{1 \leq a \leq \frac{n - 1}{2}} \left( T(a) + Ca + T(n - 1 - a) \right)\\
    \eqnote{Using $T(n-1-a) \le T(n-a)$:}\\
    &\leq &\max_{1 \leq a \leq \frac{n-1}{2}} \left( T(a) + Ca + T(n - a) \right)\\
    \eqnote{By inductive hypothesis on $T(a)$ and $T(n-a)$:}\\
    &\leq &\max_{1 \leq a \leq \frac{n-1}{2}} \left( Ca \log_2 a + Ca + C(n - a) \log_2 (n - a) \right)\\
    \eqnote{Regrouping terms:}\\
    &= &C \max_{1 \leq a \leq \frac{n-1}{2}} \left( a (\log_2 a + 1) + (n - a) \log_2 (n - a) \right)\\
    \eqnote{Using $\log_2 a + 1 = \log_2 a + \log_2 2 = \log_2 2a$:}\\
    &= &C \max_{1 \leq a \leq \frac{n-1}{2}} \left( a \log_2 2a + (n - a) \log_2 (n - a) \right)\\
    \eqnote{Using $2a < n$:}\\
    &< &C \max_{1 \leq a \leq \frac{n-1}{2}} \left( a \log_2 n + (n - a) \log_2 n \right)\\
    \eqnote{Expression under max does not depend on $a$:}\\
    &= &C ~ n \log_2 n
\end{eqnarray*}
\end{proof}
\noindent
An intuitive alternative to the above derivation may prove illuminating: in an \lstinline{App} node, we do work proportional to the smaller subtree. Let's imagine that we are touching every node in that subtree to mark the amount of work done; we need to compute the total number of touches. Now flip this around: how many times could a fixed node $v$ be touched? If we follow a path $v = u_1, u_2, ..., u_k$ from $v$ to the root, any \lstinline{App} node $u_i$ on such path could "trigger"~$v$ being touched, but only if $u_{i-1}$ was the smaller child of $u_i$. Therefore, if we follow the path $u_1, ..., u_k$, any time we see a node that triggered touching~$v$ the current subtree size at least doubles, so $v$ could only have been touched $\log n$ times. There are $n$ nodes $v$, so total touches can't exceed $n\log n$.

\begin{lemma}\label{lemma:map-ops}
Let $e$ be an expression. The total number of map operations (i.e. \lstinline{singletonVM}, \lstinline{alterVM} and \lstinline{removeFromVM}) performed by \lstinline{summariseExpr} ran on $e$ is $\mathcal{O}(|e| \log |e|)$.
\end{lemma}

\begin{proof}
We perform exactly one map operation per every \lstinline{Var} and \lstinline{Lam} node, while the total number of map operations performed in \lstinline{App} nodes is bounded due to Lemma~\ref{lemma:map-ops-app}.
\end{proof}

\begin{theorem}
Let $e$ be an expression. The total running time of the \lstinline{summariseExpr} algorithm ran on $e$ is $\mathcal{O}(|e| \log^2 |e|)$.
\end{theorem}

\begin{proof}
The \lstinline{singletonVM}, \lstinline{alterVM} and \lstinline{removeFromVM} operations are dominant, so it is sufficient to bound the time spent in these operations.

From Lemma~\ref{lemma:map-ops}, we get that the total number of these operations is $\mathcal{O}(|e| \log |e|)$. Since we implement the map as a balanced binary search tree, addition and removal take time logarithmic in terms of the size of the map (which never exceeds $|e|$), while \lstinline{singletonVM} takes constant time.
\end{proof}

\subsection{Proof That Our Hashing Function Is Strong}\label{sec:proof-hash-is-strong}

In this section, we show that \lstinline{summariseExpr} composed with \lstinline{hashESummary} is a strong hashing function. 
We assume that we have access to a \emph{source of true randomness} (e.g. a stream of random bits), which can be used to instantiate randomly chosen hash combiners. Under this assumption, we prove that it is possible to choose all hashing functions and hash combiners in a randomized way, such that the probability of hash collisions is low.

\begin{definition}\label{def:random-function}
We will call a function $f : A \rightarrow B$ \emph{random} if every value of $f(a)$ for $a \in A$ was chosen uniformly over $B$, and independently from all the other values $f(a^\prime)$ for $a^\prime \neq a$.
\end{definition}

Note that a random function $f$ according to Definition~\ref{def:random-function} is one that was \emph{chosen randomly}, but when $f$ is called for a fixed argument, its value is deterministic.

In practice, it may not be possible to obtain true randomness, or one may prefer to fix the seed and make the hashing algorithm deterministic; nevertheless, our theoretical results indicate that there is no more reason to expect hash collisions than if we had used a strong combiner in Section~\secref{fast-hash-variable-map}.

\def\HH{{\mathbb{H}}}
\def\Hk{{\HH^k}}

Throughout this section we denote the hash width as $b \in \mathbb{Z}_+$, and $\HH = \{0,1\}^b$.
By \lstinline{hash} we denote calls to a generic hash function for primitive objects.

\begin{lemma}[XOR hash combiner for sets]\label{lemma:set-xor}
Given a random function $f : A \rightarrow \HH$, define the {\em set hash} $h : 2^A \rightarrow \HH$ by
\begin{equation*}
h(S) = \bigoplus_{s \in S} f(s)
\end{equation*}
where $\bigoplus$ denotes XOR-aggregation. Then
\begin{equation*}
    \displaystyle\mathop{\bigforall}_{S_1, S_2 \subseteq A; S_1 \neq S_2} \ 
    p\left(h(S_1) = h(S_2)\right) = \frac{1}{2^b}
\end{equation*}
\end{lemma}

\begin{proof}
Fix $S_1$ and $S_2$; from the properties of XOR, we have
\begin{equation*}
    \bigoplus_{s \in S_1} f(s) = \bigoplus_{s \in S_2} f(s) \leftrightarrow \bigoplus_{s \in S_1 \ominus S_2} f(s) = 0
\end{equation*}
where $S_1 \ominus S_2$ is the symmetric difference of $S_1$ and $S_2$. As $S_1 \ominus S_2 \neq \emptyset$, we can take any $x \in S_1 \ominus S_2$, and obtain

\begin{equation}\label{eq:xor}
    \bigoplus_{s \in S_1 \ominus S_2} f(s) = 0 \leftrightarrow f(x) = \bigoplus_{s \in S_1 \ominus S_2 - \{x\}} f(s)
\end{equation}

Since the values of $f$ are chosen independently, we may assume the value for $x$ is drawn last, at which point the right side of Equation~\ref{eq:xor} is a constant. As $f(x)$ is chosen uniformly, the probability of $f(x)$ being equal to any constant is $\frac{1}{|\HH|}$.
\end{proof}

\begin{restatable}{lemma}{lemmarechash}\label{lemma:rec-hash}
Let $\mathcal{D}$ be a datatype defined recursively (such as \lstinline{Structure} or \lstinline{PosTree}). It is possible to construct in a randomized way a compositional hashing scheme $h: \mathcal{D} \rightarrow \HH$ (that is, compute the hash for $d \in \mathcal{D}$ in the constructor by calling a hash combiner on the hashes of children), so that
\begin{equation*}
    \displaystyle\mathop{\bigforall}_{a, b \in \mathcal{D}; a \neq b} \ p(h(a) = h(b)) \leq \frac{|a|+|b|}{2^b}
\end{equation*}
where $|d|$ is the number of constructor calls when building $d$ (i.e. both those for "leaf" objects and "branch" combiners).
\end{restatable}

\begin{proof}
See Appendix~\ref{appendix:lemma-rec-hash}.
\end{proof}

\begin{theorem}\label{theorem:strong-hash-pair}
Let $E$ be the set of all \lstinline{Expression} objects. It is possible to instantiate \lstinline{summariseExpr} with hashing functions and combiners into $\HH$ chosen in a randomized way, so that
\begin{equation*}
    \displaystyle\mathop{\bigforall}_{e_1, e_2 \in E; e_1 \not\equiv e_2} p(\texttt{\lstinline{h}}(e_1) = \texttt{\lstinline{h}}(e_2)) \leq 5\frac{|e_1|+|e_2|}{2^b}
\end{equation*}
where $\texttt{\lstinline{h}}(e) = \texttt{\lstinline{hashESummary}}(\texttt{\lstinline{summariseExpr}}(e))$ and $e_1 \not\equiv e_2$ means $e_1$ and $e_2$ are not $\alpha$-equivalent.
\end{theorem}

\begin{proof}
Since a full (hash-free) \esummary preserves all information relevant to $\alpha$-equivalence, the only way a collision can happen is when we convert pieces of an \esummary into hash values. We now consider all such places one by one.

First, we may get a collision when either hashing variable names, or auxiliary compositional objects (structures, position trees, and variable map entries that combine the position tree with the variable name). As the total number of calls to hash combiners in each of the four categories does not exceed $|e_1|+|e_2|$, from Lemma~\ref{lemma:rec-hash}, the probability of collision in any of the four cases is bounded by $4\frac{|e_1|+|e_2|}{2^b}$.

Moreover, collisions may arise due to XOR-aggregation of hashes when hashing variable maps. From Lemma~\ref{lemma:set-xor}, the probability of this event is bounded by $\frac{1}{2^b}$.

Finally, the top-level call to a hash combiner (to combine hashes of structure and variable map) may produce a collision. As we assume that combiner to be random, the probability of this event is simply bounded by $\frac{1}{2^b}$.

Summing up, we get
\begin{equation*}
p(\texttt{\lstinline{h}}(e_1) = \texttt{\lstinline{h}}(e_2)) \leq \frac{4(|e_1|+|e_2|) + 2}{2^b} \leq 5\frac{|e_1|+|e_2|}{2^b}
\end{equation*}
\end{proof}

\begin{theorem}\label{theorem:strong-hash-all-subexps}
Let $E$ be the set of all \lstinline{Expression} objects. It is possible to instantiate \lstinline{summariseExpr} with hashing functions and combiners into $\HH$ chosen in a randomized way, so that for any $e \in E$, \lstinline{summariseExpr} recovers the correct set of equivalence classes of subexpressions with probability at least $1 - 5|e| ^ 3 \cdot 2^{-b}$.
\end{theorem}

\begin{proof}
There are $\binom{|e|}{2} < \frac{1}{2}|e|^2$ pairs of subexpressions; any single pair can cause a collision leading to \lstinline{summariseExpr} returning an incorrect set of equivalence classes. Probability of a collision for a fixed pair of subexpressions is bounded by Theorem~\ref{theorem:strong-hash-pair}.
\end{proof}

One practical consequence of Theorem~\ref{theorem:strong-hash-all-subexps} is that 128-bit hashes are enough even for very large-scale applications. Specifically, if $b = 128$ and we consider expressions up to a billion nodes, $|e| \leq 10^9$, then the probability of having at least one collision is bounded by approximately $5 \cdot 10^{27} \cdot 2^{-128} < 10^{-10}$.
In Appendix~\ref{appendix:collisions}, we empirically verify that the observed collision rate is indeed consistent with theory.

\subsection{Incrementality}\label{sec:incremental}

One crucial property of our algorithm is compositionality: computing the hash of a subtree only requires the results from children, and there is no need to orchestrate anything across the entire expression, with the exception of ensuring all variable names are unique, which is an invariant that is easy to maintain. Therefore, the algorithm can be made incremental: if a subtree of node $v$ in an expression $e$ is modified, \esummaries (and therefore hashes) of most nodes will often stay unchanged.

More specifically, let us say that we have already computed all subtree hashes for an expression $e$, and we modify a subtree under node $v$. The only affected nodes are those lying on the path from $v$ to the root, and also those in the rewritten subtree (in particular, modifying a subtree might have required creating some fresh nodes). Recomputing \esummaries for the latter is unavoidable, and the cost of that depends on the specifics of a rewrite -- for example, if the subtree of $v$ has constant size, then the work done for the nodes in that subtree will also be constant. In this section, we will focus on the former, and try to bound the amount of work needed to recompute \esummaries for all nodes on the path from $v$ to the root.

Let $v = u_h, u_{h-1}, \cdots, u_0$ be the path in question, where $h$ is the depth of $v$. Work done when recomputing the \esummary for $u_i$ is in the worst case proportional to the size of the free variable map for $u_i$. Note that any free variable that is used in the subtree of $u_i$ is either bound in one of the nodes $u_j$ for $j < i$, or it is never bound in the entire expression $e$. If we denote the number of variables that are nowhere bound as $f$, then the work done in $u_i$ is $\mathcal{O}(i + f)$; summing over all $i$, we get a bound of $\mathcal{O}(h^2 + hf)$. Of course, the upper-bound of $\mathcal{O}(|e| (\log |e|)^2)$ still holds.

In summary, if a subtree of a node at depth $h$ is rewritten, then updating all subtree hashes takes time $\mathcal{O}(\min(h^2 + hf, |e| (\log |e|)^2))$. While in the worst case this is of the same order of magnitude as recomputing all hashes from scratch, if all variables are bound ($f = 0$), and the tree is reasonably balanced, we obtain something much faster; in particular, if the tree is balanced (i.e. has height $\mathcal{O}(\log |e|)$), then recomputing after a rewrite takes time $\mathcal{O}((\log |e|)^2)$.

\section{Empirical Evaluation}\label{sec:empirical}

\begin{table}
\centering
\caption{Algorithms considered in our evaluation. Note that some of them do not produce the correct set of equivalence classes, and are only given to define complexity minima.}
\label{tab:algorithms}
\def\Good{\textcolor{ForestGreen}{Yes}}
\def\Bad{\textcolor{red}{No}}
\begin{tabular}{llcc}
\toprule
Algorithm & Time  & True & True\\
& complexity & pos. & neg. \\
\midrule
\textcolor{col:str}{Structural} (\S\secref{syntactic-equiv}) & $\mathcal{O}(n)$ & \Good & \Bad \\
\textcolor{col:db}{De Brujin} (\S\secref{deBruijn}) & $\mathcal{O}(n \log n)$ & \Bad & \Bad \\
\midrule
\textcolor{col:ln}{Locally Nameless} (\S\secref{locally-nameless}) & $\mathcal{O}(n^2 \log n)$ & \Good & \Good \\
\textcolor{col:ours}{Ours} (\S\secref{goal} - \S\secref{step2}) & $\mathcal{O}(n ~ (\log n)^2)$ & \Good & \Good \\
\bottomrule
\end{tabular}
\end{table}

\begin{figure*}[t]
\centering
\includegraphics[width=0.48\linewidth]{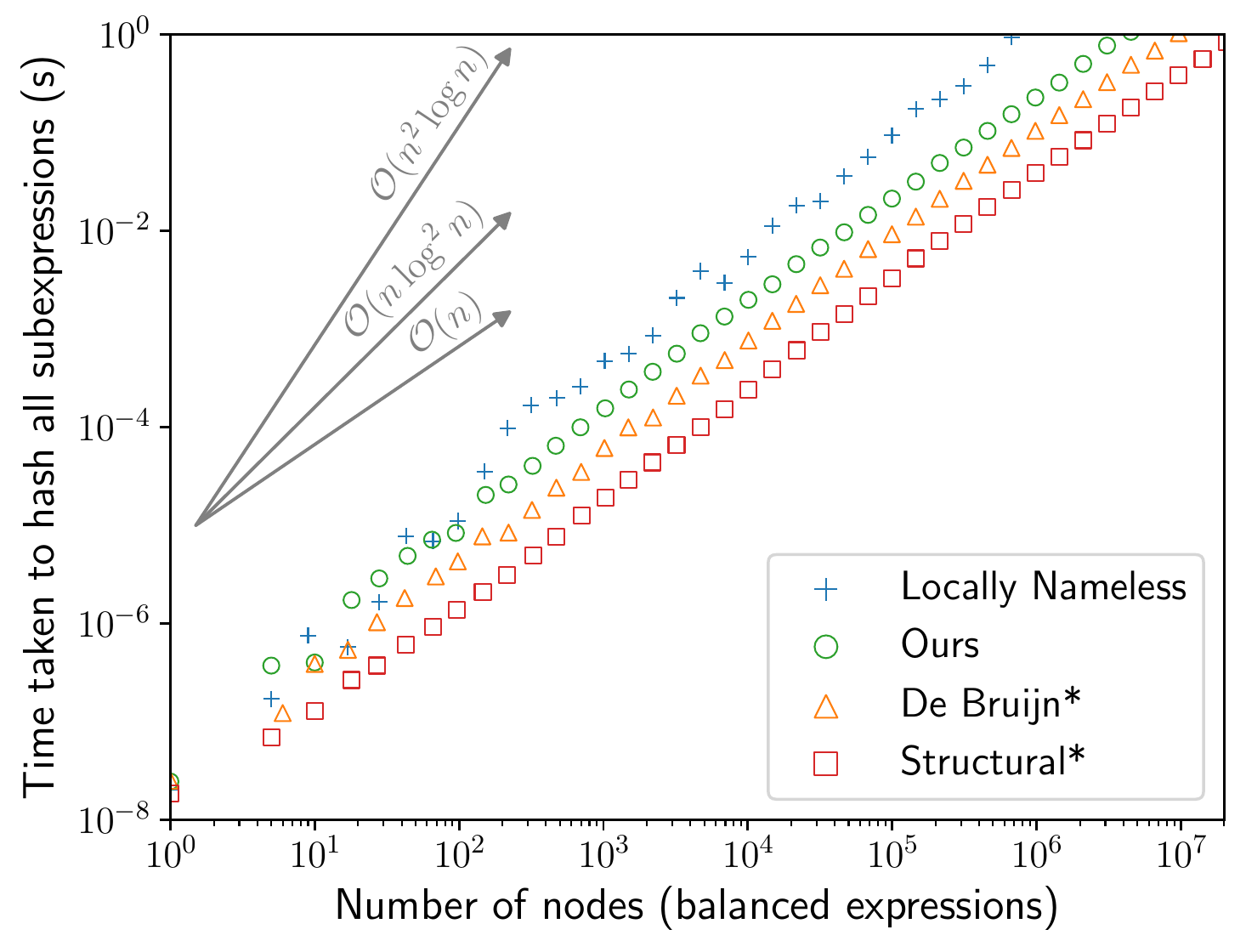}
\includegraphics[width=0.48\linewidth]{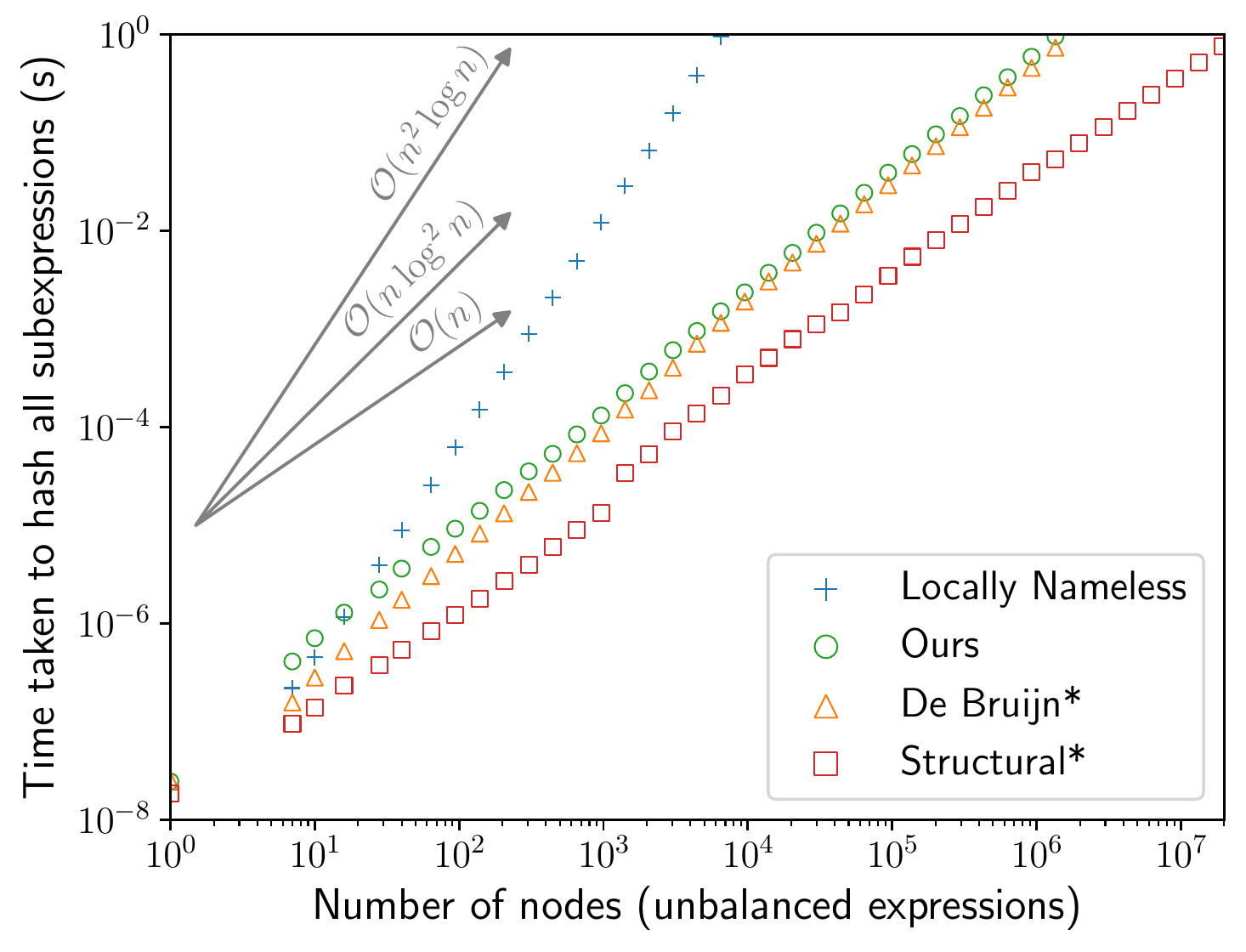}
\caption{Empirical evaluation on synthetically generated expressions: balanced trees (left), and highly unbalanced ones (right). Note that the algorithms marked with (*) produce an incorrect set of equivalence classes, so the key comparison is between Locally Nameless and Ours.}
\label{fig:random-exprs-benchmark}
\end{figure*}

In this section we empirically evaluate the running time of our final algorithm. We consider two settings: synthetic, automatically generated lambda terms (Section~\secref{eval-synthetic}), and several hand-picked realistic examples corresponding to commonly used machine learning models (Section~\secref{eval-hand-picked}).

In Table~\ref{tab:algorithms}, we list the hashing algorithms that we compare in this section. The first two,
Structural and De Bruijn, are incorrect: they do not meet the specification outlined in Section~\secref{goal}. 
Specifically, they may equate distinct 
expressions (false positives), or fail to equate $\alpha$-equivalent ones (false negatives).  We present these algorithms here to 
give a sense of the extra performance cost of hashing modulo $\alpha$-equivalence.

The Locally Nameless algorithm is the fastest one we know that meets the specification, while Ours is
the algorithm presented in this paper. 

We implemented all four in Haskell, over the following expression type
\begin{code}
data Expression h = Var h Name
                  | Lam h Name Expression
                  | App h Expression Expression
\end{code}

In each case, the hashing algorithm simply annotates each node with a hash-value, yielding a result of type \lstinline|Expression| \lstinline|HashCode|. We did not model the cost of putting these hash-codes into a hash table and identifying equivalence classes, since this cost is the same in all cases. The implementation of our algorithm is optimized over the source code listed in this paper by the addition of strictness annotations and replacing two map operations with a single fused map operation in a couple of places; similar optimizations were applied to the baseline algorithms.
The garbage collector was disabled during timing.
Constant factors may of course vary in other implementations, but we are mainly interested in how the algorithms behave relative to each other.

\subsection{Random Expressions} \label{sec:eval-synthetic}

In Figure~\ref{fig:random-exprs-benchmark}, we show time taken by the four algorithms to hash all subexpressions of randomly generated expressions of varying size.  We generated two different families of random expressions:
\begin{itemize}
    \item \emph{Balanced trees}.  Here we generated expressions that are roughly balanced trees, at each point generating a \lstinline|Lam| or \lstinline|App| node with equal probability.  Each \lstinline|Lam| node has a fresh binder, and at variable occurrences we choose one of the in-scope bound variables.
    \item \emph{Wildly unbalanced trees} with very deeply nested lambdas.  This case is not as unrealistic as it sounds:
    a realistic language will include \lstinline|let| bindings, and deeply-nested stacks of \lstinline|let| expressions are very common in practice, especially in machine-generated code.
\end{itemize}
\noindent
The results reassure us that our algorithm meets the claimed complexity bounds -- note the quadratic behaviour of Locally Nameless for unbalanced trees.  Moreover, although there is a constant-factor cost compared to the non $\alpha$-respecting algorithms, the slowdown is much smaller than Locally Nameless.

\subsection{Real-Life Examples} \label{sec:eval-hand-picked}

Our interest in the problem of hashing modulo $\alpha$-equivalence was directly motivated by our parallel work on a prototype compiler for machine learning models, and the discovery that a significant amount of time was being spent on pre-processing the AST to find equivalent subtrees. In that context, in Table~\ref{tab:realistic-exprs-benchmark} we show a different empirical evaluation, using real expressions found in machine learning workflows: "MNIST CNN"~\cite{le1989handwritten} is a convolution kernel from a deep neural network used in computer vision; "GMM"~\cite{reynolds2009gaussian} is the Gaussian Mixture Model benchmark from the ADBench suite~\cite{ADBench}; and "BERT"~\cite{devlin2018bert} is a model from natural language processing, implemented using the PyTorch library~\cite{PyTorch}. Conveniently for our benchmarks, BERT also has a parameter controlling the number of ``layers'', which linearly scales the expression size due to loop unrolling. We see that on practical examples, our algorithm is only up to $4\times$ slower than running de Brujin on the expression once, and much faster than the locally nameless baseline, while enjoying a better bound on the worst-case time complexity.  In Figure~\ref{fig:bert-benchmark} we show performance on BERT as the number of layers is varied.

\begin{table}
\centering
\caption{Empirical evaluation on hand-picked realistic expressions. Each measurement is time in milliseconds to compute all subexpression hashes for each expression. Note that the algorithms marked with (*) produce an incorrect set of equivalence classes.}
\label{tab:realistic-exprs-benchmark}
\def\tm#1#2{$#1$ ms}
\def\colhead#1{#1}
\begin{tabular}{lrrr}
\toprule
Algorithm & \colhead{MNIST CNN} & \colhead{GMM} & \colhead{BERT 12}\\
          & \colhead{$n=840$} & \colhead{$n=1810$} & \colhead{$n=12975$}\\
\midrule
\textcolor{col:str}{Structural*} & \tm{0.011}{??} & \tm{0.027}{??} & \tm{0.38}{??} \\
\textcolor{col:db}{De Bruijn*} & \tm{0.035}{??} & \tm{0.089}{??} & \tm{1.70}{??} \\
\midrule
\textcolor{col:ln}{Locally Nameless} & \tm{0.30}{??} & \tm{2.00}{??} & \tm{820.0}{??} \\
\textcolor{col:ours}{Ours} & \tm{0.14}{??} & \tm{0.36}{??} & \tm{3.6}{??} \\
\bottomrule
\end{tabular}
\end{table}

\section{Related Work}\label{sec:related}

We found surprisingly few papers about the problem of finding common subexpressions modulo $\alpha$-equivalence.
Shao \emph{et al.} describe the impressive FLINT compiler, which uses de Bruijn representation and aggressive hash-consing to achieve very compact type representations and constant-time equality comparison \cite{flint}. It is not clear how they deal with the false positives and false negatives we mention in Section~\secref{deBruijn}. Murphy takes a similar approach in the TILT compiler \cite{murphy:tilt}.
Again, the goal is structure sharing and the mechanism is de Bruijn indexing, but he seeks to conceal the tiresome de Bruijn index-shuffling (which is somewhat exposed in FLINT) behind an abstraction ``curtain'' that allows the client to use a simpler named interface. He mentions the problem of false negatives, and concludes that the overheads of his abstractions are too high.

\begin{figure}[t]
\centering
\includegraphics[width=0.96\linewidth]{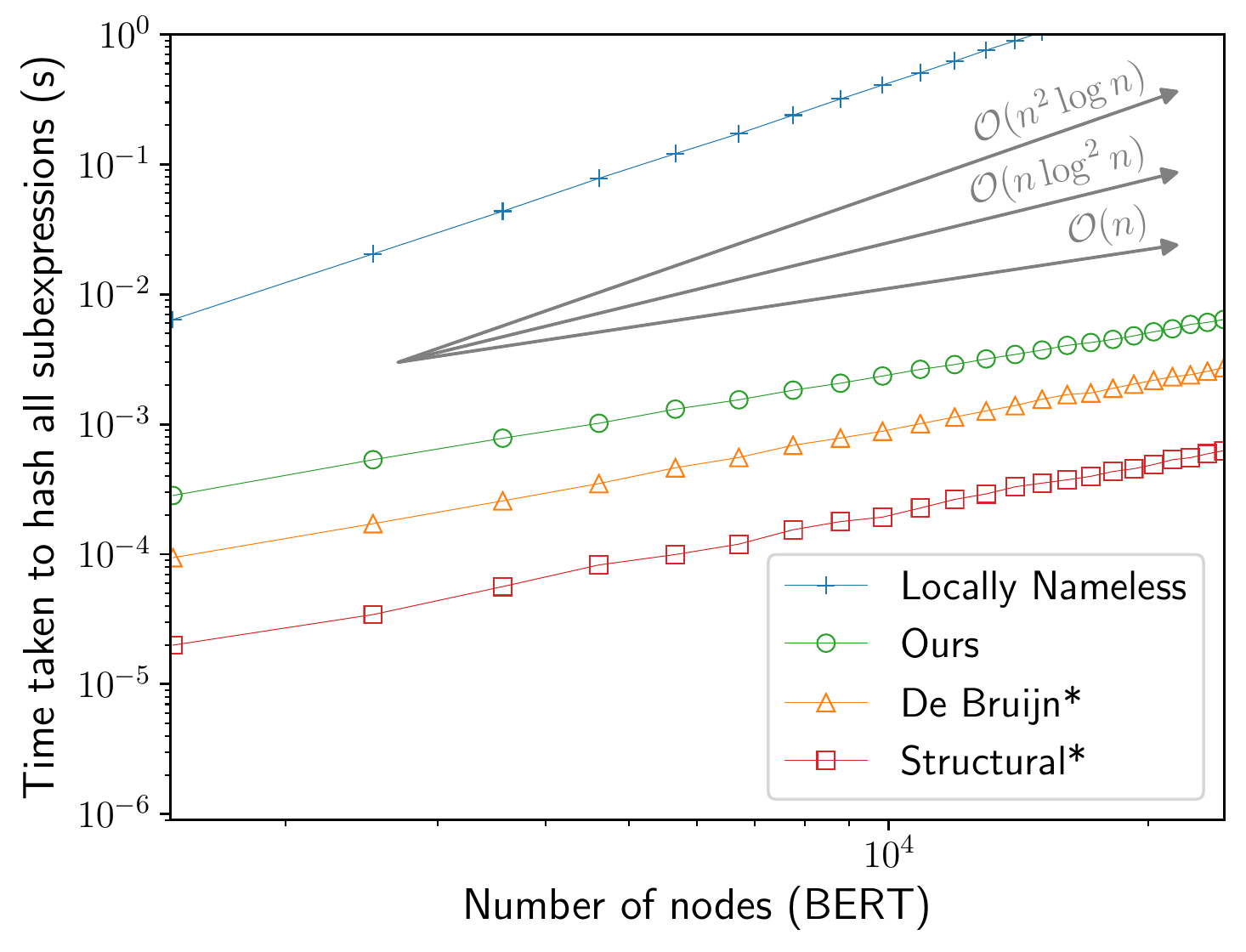}
\caption{Empirical evaluation on expressions from the BERT model.  Note that the algorithms marked with (*) produce an incorrect set of equivalence classes, so the key comparison is between Locally Nameless and Ours.}
\label{fig:bert-benchmark}
\end{figure}

Filliatre and Conchon describe a hash-consing library in OCaml, again with the goal of structure sharing \cite{filliatre}.  But their focus is very different to ours: they are concerned about the API design for a hash-consing library, including issues such as when to clear out the hash table. Concerning $\alpha$-equivalence, they use de Bruijn indexing from the outset, without discussion.

The ``locally nameless'' representation~\cite{mcbride04,chargueraud12} has a long history, indeed Weirich {\em et al.}~\cite{weirich} observe that it ``is mentioned in the conclusion of de Bruijn’s paper''. They further note that "If we remove names from bound patterns (which are preserved only for error messages) the locally nameless representation interacts nicely with hash-consing, as all $\alpha$-equivalent terms have the same representation".

The idea of representing a lambda with a ``map'' of the occurrences of its bound variable, which we adopt for our \esummaries in Section~\secref{step1},
has been studied before~\cite{Abel_2011,sato_2013}. Kennaway and Sleep describe another representation, director strings, in which information
about occurrences is stored in the application nodes, rather than the lambdas \cite{director-strings}.  McBride gives a very helpful overview of these approaches \cite{McBride_2018}, but none of them addresses the question of compositional hashing.

Dietrich \emph{et al.} \cite{dietrich} discuss hashing source code abstract syntax trees (ASTs), with the goal of minimising rebuilds in a build system. If, for example, one changes the white-space layout, the timestamp of the file will change, but the AST (and its hash) will not. They do not consider $\alpha$-equivalence at all, and it seems likely that an $\alpha$-renamed program would indeed be considered different by their system.

There is literature~\cite{zhao,thomsen12} on detecting code clones, or plagiarisms, which does typically hash ASTs, but there the goal is usually to generate many pairs of candidates (i.e.\ false positives are welcomed), so does not apply to our use cases.

\paragraph{Acknowledgements}  We warmly thank these colleagues for their feedback on earlier drafts of this paper: Conor McBride, David Collier, Leonardo de Moura, Max Willsey, Stephanie Weirich and Tom Minka.


\balance

\clearpage

\appendix

\section{Proof of Lemma \ref{lemma:rec-hash}}\label{appendix:lemma-rec-hash}

For convenience, we first restate Lemma~\ref{lemma:rec-hash} below, and then outline the proof.

\lemmarechash*

\begin{proof}
Given $d = Con(d_1, \cdots, d_k) \in \mathcal{D}$, we define
\begin{equation*}
h(d) = f(|d|, hash(Con), h(d_1), \cdots, h(d_k))
\end{equation*}
where $f$ is a random hash combiner. That is, we combine the hashes of children and the constructor, and salt it with the size $|d|$ of the object $d$. As $f$ only accepts elements of $\HH$ as arguments, here we silently assume $|d| < 2^b$; if that is not the case, then the bound to be proven is vacuous, and hence there is nothing to do.

We will now bound the probability of $h(a) = h(b)$ by induction on $\max(|a|, |b|)$.

It is easy to check that for $|a| = |b| = 1$ (i.e. both $a$ and $b$ are leaves) the inequality holds. Now, assume $\max(|a|, |b|) \geq 2$; without loss of generality $|a| \geq |b|$.

Denote $a = Con_a(a_1, \cdots)$, $b = Con_b(b_1, \cdots)$. We now distinguish three cases.

\paragraph{Case 1: $|a| > |b|$}

We can see that computing $h(a)$ and $h(b)$ involves exactly one call to $f$ of the form $f(|a|, *)$, while all the other calls are for $f(x, *)$ for $x < |a|$. Since the values for $f$ are drawn uniformly and independently, we can assume they are drawn in the order of increasing first argument to $f$ (breaking ties arbitrarily). When the value for $f(|a|, *)$ is being drawn, all other values can be considered constant, and so $p(h(a) = h(b)) = \frac{1}{2^b}$.

\paragraph{Case 2: $|a| = |b|$, $Con_a \neq Con_b$}

Here $h(a) = h(b)$ requires that either $hash(Con_a) = hash(Con_b)$, or the top level calls to $f$ produce a collision. Similarly to the previous case, the probability of either of these events can be bounded by $\frac{1}{2^b}$, and we get $p(h(a) = h(b)) \leq \frac{2}{2^b}$.

\paragraph{Case 3: $|a| = |b|$, $Con_a = Con_b$}

In this final case, the first two arguments to $f$ are the same, which means that a hash collision on children of $a$ and $b$ can possibly propagate upwards and imply a collision at the top level. Therefore, we will need to use the inductive hypothesis.

More specifically, $h(a) = h(b)$ can arise in two ways: either $h(a_i) = h(b_i)$ for all $i$, which of course implies $h(a) = h(b)$, or the two tuples of arguments to $f$ do not match, and the collision is produced with the two top-level calls to $f$.

For the former case, recall that $a \neq b$, and therefore $a_i \neq b_i$ for some $i$. Combining that with $h(a_i) = h(b_i)$, we can apply the inductive hypothesis; this is legal, since $\max(|a_i|, |b_i|) < \max(|a|, |b|)$. The probability of $h(a_i) = h(b_i)$ is therefore bounded by $\frac{|a_i| + |b_i|}{2^b}$. The case when the top level calls to $f$ produce a collision can be analyzed as before, yielding probability of collision equal to $\frac{1}{2^b}$.

Summing up the probabilities from the two subcases, we get

\begin{equation*}
    p(h(a) = h(b)) \leq \frac{|a_i| + |b_i| + 1}{2^b} < \frac{|a| + |b|}{2^b}
\end{equation*}

\end{proof}

\section{Empirical Frequency of Hash Collisions}\label{appendix:collisions}

To experimentally verify the bound from Theorem~\ref{theorem:strong-hash-pair}, in this section we evaluate the empirical frequency of hash collisions. To do this, we first modified our algorithm to use 16-bit integers, as for 32-bit and above one needs an enormous number of trials to find a collision.

However, measuring the amount of collisions in a meaningful way is non-trivial: while we can check how often hashes of two \emph{random} expressions collide, it might be that more collisions arise in real applications; what is worse, there may be adversarial pairs of expressions specially crafted to make our hashes collide. Of course, if the hash function is fixed, there \emph{exist} pairs of expressions that produce a hash collision. However, note that Theorem~\ref{theorem:strong-hash-pair} assumes our hash function is \emph{not} fixed: instead, the hash combiners that are used should be chosen randomly. In practical terms, our theorem states the following: if we instantiate our algorithm, and seed its hash combiners with a randomly chosen seed, then \emph{there is no way to consistently break it} - while for a \emph{fixed} seed one can laboriously find a collision, there is no pair of expressions that would collide reliably across many seeds. This is a much stronger claim than just stating that two random expressions rarely collide, and to verify it, one needs to play a malicious user, and try to construct adversarial examples that are more likely to collide than random ones.

Therefore, in this section we consider two ways to create a pair of expressions:

\begin{itemize}
\item \emph{Random expressions}.  Here we generate two random balanced expressions as in Section~\secref{eval-synthetic}, and discard pairs that turn out to be alpha-equivalent.
\item \emph{Adversarial expressions}.
As discussed above, we acted as an adversary, and designed a way to generate pairs of expressions that are more likely to produce collisions. We describe this in detail in Appendix~\secref{adversarial-exprs}.
\end{itemize}

For both ways of generating pairs of expressions we varied the expression size between $128$ and $4096$, for each size drawing $10 \cdot 2^{16}$ pairs and hashing them looking for collisions. We then divided the resulting number of collisions by $10$ to get an estimated number of collisions per $2^{16}$ samples. Note that the resulting value for a perfect hash function would be $1$ (in expectation). On the other hand, Theorem~\ref{theorem:strong-hash-pair} upper-bounds the collision probability by $5 \frac{|e_1|+|e_2|}{2^b}$, and since we are comparing pairs of expressions of the same size $|e_1| = |e_2| = n$, we get an upper-bound of $10n$ collisions per $2^{16}$ samples.

\begin{figure}
\centering
\includegraphics[width=\linewidth]{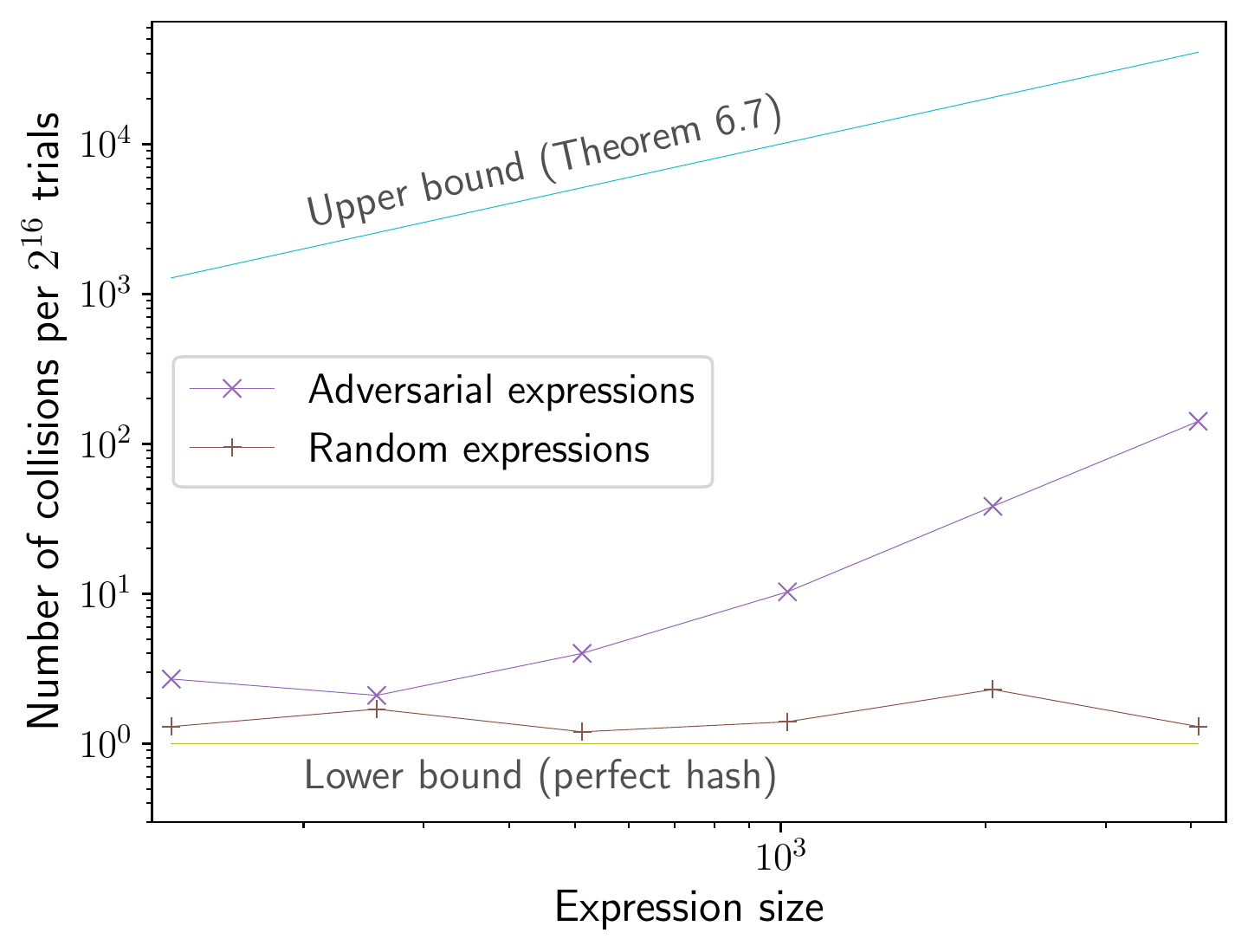}
\caption{Empirical number of collisions for both random and adversarial pairs of expressions with varying size.}
\label{fig:collisions}
\end{figure}

In Figure~\ref{fig:collisions}, we plot the resulting number of collisions for the two ways of generating expression pairs, as well as a lower-bound (perfect hash function) and upper-bound (Theorem~\ref{theorem:strong-hash-pair}). For random expressions, our hash achieves a close-to-perfect number of collisions, which does not appear to grow with $n$. On the other hand, adversarial expression pairs generate more collisions as $n$ grows, but still two orders of magnitude less than the theoretical upper-bound. Note that we do not consider $n > 4096$, as for $n = \frac{2^{16}}{10} \approx 6500$ our upper-bound becomes vacuous: $n$ is too close to $2^b$ to provide any guarantees on the frequency of collisions.

\subsection{Generating Adversarial Expression Pairs}\label{sec:adversarial-exprs}

Here we describe the procedure to generate adversarial pairs of expressions. This process is not specialized to our specific algorithm, and hence may work for other compositional hashing algorithms that act on tree-like objects.

The idea is as follows: we start with two small non-alpha-equivalent expressions with no free variables. Concretely, we choose
\begin{code}
$e_1$ = \x . App (Var x) (App (Var x) (Var x))
$e_2$ = \x . App (App (Var x) (Var x)) (Var x)
\end{code}
Then, until the right expression size is reached, we transform the expressions by wrapping both of them in either a \lstinline{Lam} or an \lstinline{App} node. In other words, we create a pair of highly unbalanced expressions (similarly to Section~\secref{eval-synthetic}) which differ only at the very bottom. Intuitively, when hashing the resulting expressions, the (likely different) hashes of $e_1$ and $e_2$ will get repeatedly transformed \emph{in the same way} when the algorithm computes the hashes for larger subtrees. Crucially, if the hashes collide at some point, they will stay the same indefinitely, as the way $e_1$ and $e_2$ are extended upwards is the same. Hence, a collision at the lower level will propagate to cause a collision at the top level, causing the collision probability to grow with expression size.

\section{An Alternative to StructureTag}
\label{appendix:number-theory}

Our initial algorithm of Section~\secref{full-algo-quadratic} transforms \lstinline{PosTree}s from both subtrees of an \lstinline{App} node. Since this is prohibitively expensive, in Section~\secref{smaller-tree} we show that it is enough to transform only \emph{one}
 of the subtrees, as long as we introduce an appropriately chosen \lstinline{StructureTag}. In this section, we discuss an alternative to the \lstinline{StructureTag} approach, which yields an algorithm with the same final time complexity.
 
Since this alternative is more complex, we refrain from mentioning it in the main body of the paper. Here we describe the key ideas involved, which the reader may find interesting.

To arrive at the alternative formulation, consider the following question: can we transform \lstinline{PosTree}s from \emph{both} children, and still get good time complexity due to laziness? In other words, can we avoid actually tagging \lstinline{PosTree}s in a given map, and instead lazily store the transformation to be applied?

Formally, consider a variable map after the optimizations of Section~\secref{hash-structure}, i.e. with \lstinline{PosTree}s simply being represented by hash-codes (elements of $\HH$). Extending a \lstinline{PosTree} with one of the markers (\lstinline{PTLeftOnly}, \lstinline{PTRightOnly} and \lstinline{PTBoth}) corresponds to transforming the hash-code in a fixed way (although dependent on the particular hash combiner); denote these transformations as $f_L, f_R : \HH \rightarrow \HH$ and $f_{both} : \HH^2 \rightarrow \HH$. To simplify the following analysis, we will now focus on $f_L$ and $f_R$; extending it to handle $f_{both}$ is easy, as in a given node it only needs to be called at most as many times as the size of the \emph{smaller} of the children's variable maps.

If we focus on the set of values in the variable map, the problem is the following: we want to maintain a set of hash-codes, with the possibility of quickly applying either $f_L$ or $f_R$ on all the values. This hints at a lazy solution: maintain a transformation $f : \HH \rightarrow \HH$ together with the set of hash-codes, with the meaning that $f$ should be applied on all elements of the set. Applying $f_L$ on all elements of such a lazy-transformation-augmented set is just a matter of setting $f^\prime = f_L \circ f$.

However, this is not so easy: when we look up an entry in the variable map, we need to pass the obtained value through the lazy transformation $f$. Therefore, $f$ has to be represented in a way such that it is possible to evaluate it in constant time, even though $f$ may have been created out of a very long sequence of function compositions. Moreover, adding a new entry to the variable map requires passing the newly added value $x$ through $f^{-1}$, so that when the value is read out later and passed through $f$ we recover $f(f^{-1}(x)) = x$.

It remains to show an efficient representation of functions $f : \HH \rightarrow \HH$, such that composing, evaluating, and inverting takes constant time.

To simplify this, we can notice that fast inversion is not strictly necessary if we always manipulate pairs of a function and its inverse i.e. $(f, f^{-1})$; note that composing $(f, f^{-1})$ with $(g, g^{-1})$ can be computed as $(f \circ g, g^{-1} \circ f^{-1})$. We still need to be able to invert our fundamental building blocks ($f_L$ and $f_R$), but since that happens only once before the algorithm commences, it does not have to be done in constant time.

While many possible representations that satisfy the aforementioned requirements exist, one natural choice are linear functions, i.e. functions of the form $f(x) = a \cdot x + b$, where $a, b \in \HH$, and all operations are carried out modulo $|\HH|$. Indeed, a linear $f$ can be represented with just a pair of $(a, b)$, evaluating it on $x \in \HH$ takes constant time, and a composition of two linear functions represented by $(a_f, b_f)$ and $(a_g, b_g)$ is $(a_f \cdot a_g, a_f \cdot b_g + b_f)$. Invertibility can be guaranteed by requiring that $a$ is coprime with $|\HH|$; for the case of $\HH = \{0,1\}^b$, this simply means that $a$ is odd.

Using linear transformations on hash-codes poses a potential risk, as it could lead to collisions, especially if the choice of $f_L$ and $f_R$ is particularly unfortunate; because of that, using a \lstinline{StructureTag}-based variant is preferable. However, we have also implemented the variant described in this section, and found that it in practice it also produces strong hashes. We believe that, with careful analysis, one could likely derive guarantees similar to the one of Theorem~\ref{theorem:strong-hash-pair}.

\end{document}